\documentclass{scrartcl}

\RequirePackage{amsthm,amsmath,amsfonts,amssymb}
\RequirePackage[colorlinks,citecolor=blue,urlcolor=blue]{hyperref}
\RequirePackage{graphicx}

\usepackage{graphicx}

\newcommand{\id}{\texttt{id }}

\usepackage{amsfonts}
\usepackage{placeins}
\usepackage{bm}
\usepackage{appendix}
\usepackage{chngcntr}
\usepackage{etoolbox}
\usepackage{url} 
\usepackage{natbib}
\usepackage{booktabs}
\usepackage{xcolor}
\usepackage[]{algorithm2e}
\usepackage{physics}
\usepackage{amsthm}
\usepackage[final]{pdfpages}
\usepackage{dsfont}
\usepackage{comment}
\usepackage{color}
\usepackage{soul}
\usepackage{latexsym}
\usepackage{mathrsfs}
\usepackage{rotating}
\allowdisplaybreaks
\usepackage{cancel}
\usepackage[utf8]{inputenc}
\usepackage{csquotes}
\usepackage{bbm}
\usepackage{verbatim}
\usepackage{amsmath}
\usepackage{tcolorbox}
\usepackage{longtable}
\usepackage{afterpage}
\usepackage{pdflscape}
\usepackage{capt-of}
\usepackage{lscape}
\usepackage{mathtools}
\usepackage[maxfloats=25]{morefloats}
\usepackage[export]{adjustbox}
\usepackage{multirow}
\usepackage{chngpage}
\usepackage{blkarray, bigstrut}
\usepackage{xparse}
\usepackage[mathscr]{eucal}
\usepackage{tikz}
\newtheorem*{remark}{Remark}

\renewcommand{\id}{\texttt{id}}

\usetikzlibrary{fit,positioning}
\tikzset{
	treenode/.style = {shape=rectangle, rounded corners,
		draw, align=center,
		top color=white,
		bottom color=blue!20},
	root/.style     = {treenode, font=\Large,
		bottom color=blue!20},
	env/.style      = {treenode, font=\Large},
	dummy/.style    = {circle, font=\Large,draw,
		bottom color=blue!20,top color=white}
}

\newtheorem{theorem}{Theorem}[section]
\newtheorem{lemma}[theorem]{Lemma}
\theoremstyle{remark}

\title{Distributional Results for Model-Based Intrinsic Dimension Estimators}

\author{
	Francesco Denti\footnote{Department of Statistics, University of California, Irvine, United States}\\
	\texttt{fdenti@uci.edu}
	\and
	Diego Doimo\footnote{SISSA, Via Bonomea 265, Trieste, Italy} \\
	\texttt{ddoimo@sissa.it} 
	\and
	Alessandro Laio\footnotemark[2]\\
	\texttt{laio@sissa.it}
	\and
	Antonietta Mira\footnote{Faculty of Economics,
		Universit\`a della Svizzera italiana, Lugano, Switzerland, and
		Insubria University, Varese, Italy}\\
	\texttt{antonietta.mira@usi.ch}
}

\date{}

\begin{document}
	\maketitle
	\begin{abstract}
		Modern datasets are characterized by a large number of features that describe complex dependency structures. To deal with this type of data, dimensionality reduction techniques are essential. Numerous dimensionality reduction methods rely on the concept of intrinsic dimension, a measure of the complexity of the dataset. In this article, we first review the \texttt{TWO-NN} model, a likelihood-based intrinsic dimension estimator recently introduced by \cite{Facco}. Specifically, the \texttt{TWO-NN} estimator is based on the statistical properties of the ratio of the distances between a point and its first two nearest neighbors.
		We extend the \texttt{TWO-NN} theoretical framework by providing novel distributional results of consecutive and generic ratios of distances. These distributional results are then employed to derive intrinsic dimension estimators, called \texttt{Cride} and \texttt{Gride}. These novel estimators are more robust to noisy measurements than the \texttt{TWO-NN} and allow the study of the evolution of the intrinsic dimension as a function of the scale of the distances. We discuss the properties of the different estimators with the help of simulation scenarios.

	\end{abstract}

	\section{Introduction}
	\label{sec:intro}
	
	In recent years, we have witnessed an unimaginable growth in data production. From personalized medicine to finance,  datasets characterized by large dimensions are ubiquitous in modern data analyses.
	The availability of these high-dimensional datasets poses novel and engaging challenges for the statistical community, called to devise new techniques to extract meaningful information from the data in a reasonable amount of time. Fortunately, data that are contained in high-dimensional embeddings can often be described by a handful of variables: 
	a subset of the original ones or a combinations - not necessarily linear - thereof.
	In other words, one can effectively map the features of a dataset onto spaces of much lower dimension, such as nonlinear manifolds \citep{Levina}. 
	Estimating the dimensionality of these latent manifolds is of paramount importance. 
	We will call this quantity of interest the intrinsic dimension (\texttt{id} from now on) of a dataset, i.e., the number of relevant coordinates needed to accurately describe the data-generating process.\\
	Many other definitions of \texttt{id} have been proposed in the literature. For example, \citet{Fukanaga1972} described the \texttt{id} as the minimum number of parameters needed to accurately describe the important characteristics of a system. 
	For \citet{Bishop95}, the \texttt{id} is the dimension of the subspace where the data lie entirely, without information loss. Alternatively, \citet{Campadelli2015} provided another useful interpretation of the \texttt{id} within the pattern recognition literature. In this case, a set of points is viewed as a sample uniformly generated from a distribution over an unknown smooth (or locally smooth) manifold structure (its support),
	eventually embedded in a higher-dimensional space through a non-linear smooth mapping. Then, the \texttt{id} to be estimated is the manifold's topological dimension.\\ 
	All these definitions are useful to delineate different aspects of the multi-faceted concept that is the \texttt{id}.\\
	
	The literature regarding statistical methods for dimensionality reduction and \texttt{id} estimation is extremely vast and heterogeneous. We refer to \citet{tesiFacco,Campadelli2015} for comprehensive reviews. Generally, methods for the estimation of the $\id$ can be divided into two main families: \emph{projective methods} and \emph{geometric methods}.
	On the one hand, \emph{projective methods} estimate the low-dimensional embedding of interest through transformations of the data, which can be linear or nonlinear. Famous members of this family are the traditional Principal Component Analysis (PCA) \citep{Hotelling1933} and the Multidimensional Scaling \citep{Pigden1988}.
	In both cases, the goal is to find the best linear projection of the data, with respect to some pre-specified loss function, onto a lower dimensional space. However, many data manifolds cannot be described by a simple linear combination of the features in a dataset.
	Thus, several authors focused on the development of nonlinear algorithms such as Local Linear Embedding \citep{Roweis2000}, the Isomap \citep{Tenn}, and others \citep{Belkin2002,Donoho2003}. See also \citet{Jollife2016} and the references therein.\\
	On the other hand, \emph{geometric methods} rely on the topology of a dataset, exploiting the properties of the distances between data points. Within this family, we can distinguish between \emph{fractal methods}, \emph{graphical methods}, and \emph{methods based on nearest neighbor distances}.\\
	The first class focuses on how the number of neighbors of a given point increases while increasing the dimension of its neighborhood. The concept at the basis of all fractal methods is that the volume of a $d$-dimensional ball of radius $r$ scales as $r^{d}$ \citep{Falconer2003}. Thus, these estimators are based on the idea of counting the number of observations in a neighborhood of radius $r$ to estimate its rate of growth $\hat{r}$. Since the estimated growth is assumed to resemble the theoretical growth rate $r^{d},$ these methods exploit the connection between the empirical $\hat{r}$ and $r^d$ to estimate the parameter $d$, regarded as the fractal dimension of the data.\\ 
	Theory and algorithms for \emph{graphs} can also be exploited to estimate the \texttt{id} of datasets. In particular, graph theory is especially useful when dealing with non-linear subspaces. A graph obtained by linking points close to each other can provide valuable insights regarding the geometry of the latent manifold where the data are supposed to lie and the \emph{geodesic distance} represents a reliable distance measure in this context. This type of distance is ``shape-aware'', i.e. capable of measuring the length of paths contained in the manifold and to analyze the scaling behavior of the distance probability distribution at intermediate length-scales. 
	For example, to capture the non-linearity of the subspace, i.e. to perform manifold learning,
	\citet{Granata2016} provided a method to estimate a global \texttt{id} starting from the distribution of the geodesic distance. \citet{Costa2004} recovered the geodesic minimum spanning tree and the \texttt{id} $d$ of the dataset via a linking equation. These are example of how exploiting a graph structure, built connecting neighboring points, allows to uncover involved topological properties impossible to recover within the classical euclidean framework. \\
	\emph{Nearest neighbors (NNs) methods} rely on the assumption that points close to each other are uniformly drawn from $d$-dimensional balls (hyper-spheres). More formally, consider a generic data point $\bm{x}$ and denote with $\mathcal{B}_{d}(\bm{x}, r)$ a hyper-sphere, characterized by small radius $r \in \mathbb{R}^{+}$, 
	centered in point $\bm{x}$. If $\rho(\bm{x})$ is a density distribution defined on $\mathbb{R}^{d}$, the following approximation holds: $\frac{k}{n} \approx \rho(\bm{x})\, \omega_{d}\, r^{d} $, where $k$ is the number of NNs of $\bm{x}$ within the hyper-sphere $\mathcal{B}_d(\bm{x}, r)$, while $\omega_{d}$ is the volume of the $d$-dimensional unit hyper-sphere in $\mathbb{R}^{d}$ \citep{Pettis1979}. Intuitively this tells that the proportion of points of a given sample which fall into the ball $\mathcal{B}(\bm{x}, r)$ is approximately $\rho(\bm{x})$ times the volume of the ball. If the density is constant, one can estimate the \texttt{id} using only the average distances from a point's $k$ NNs. \\
	From a different perspective, some authors adopted modeling frameworks for manifold learning and \texttt{id} estimation that are based on a probabilistic distribution for the distances between data points. \citet{Amsaleg2015}, exploiting results from \citet{Houle2013}, suggested modeling a distance random variable using a Generalized Pareto Distribution \citep{Coles2008} since they showed that the \texttt{id} can be recovered, asymptotically, as a function of its parameter. 
	Additionally, some model-based methods to explore the topology of datasets have recently been developed, pioneered by the likelihood approach discussed in \citet{Levina}. Recently, \citet{Duan} proposed to model the pairwise distances among distributions to coherently estimate a 
	clustering structure in a Bayesian setting. One drawback of this method is that it involves the computation of each pairwise distance among the data points, which can be extremely computationally expensive. \citet{Mukhopadhyay2019} used Fisher-Gaussian kernels to estimate densities of data embedded in non-linear subspaces. \citet{Li2017} proposed to learn the structure of latent manifolds by approximating them with spherelets instead of locally linear approximation, developing a spherical version of PCA. In the same spirit, \citet{Li2019} applied this idea to the classification of data lying on complex, non-linear, overlapping and intersecting supports. Similarly, \citet{Li2019a} proposed to use the spherical PCA to estimate a geodesic distance matrix between the data, which takes into account the structure of the latent embedding manifolds and create a spherical version of the $k$-medoids algorithm \citep{Kaufman1987}.\\ 
	
	In this paper, we introduce and discuss novel likelihood-based approaches for the $\id$ estimation that stem from the geometrical properties of NNs. 
	Specifically, we build on the work of \citet{Facco}, where the authors proposed the two nearest neighbors (\texttt{TWO-NN}) estimator. The \texttt{TWO-NN} is a model-based \texttt{id} estimator derived from the properties of a Poisson point process, whose realizations take place in a manifold of dimension $d$. 
	They proved that the ratio of distances between the second and first NNs of a given point is Pareto distributed with unitary scale parameter and shape parameter precisely equal to $d$. Their result holds under mild assumptions on the data-generating process, that we will discuss in detail in the following. Therefore, they suggested estimating the \texttt{id} by fitting a Pareto distribution to a proper transformation of the data. \\
	
	Our contribution is twofold. First, while introducing the modeling setting, we revisit the main results presented in \citet{Facco}. In particular, we provide alternative proofs for the 
	validity of the \texttt{TWO-NN} estimator 
	by using standard properties of random variable distributions. 
	Moreover, we also present the maximum likelihood and Bayesian counterparts of the \texttt{TWO-NN} estimator.\\ 
	Second, we extend the \texttt{TWO-NN} theoretical framework by deriving closed-form distributions for the product of consecutive ratios of distances and, more importantly, for the ratio of distances between NNs of generic order. In addition to the contribution to the Poisson process theory, our extensions have relevant practical consequences. Indeed, considering ratios beyond the second order allows the investigation of the $\id$ evolution as a function of the distances between NNs. In other words, we employ our modeling extensions to study how the estimate is sensitive to scale effects. Considering the evolution of the estimates as the scale changes allows us to obtain an $\id$ estimator that is more robust to noise in the data.  With the help of a simulation study, we discuss how these results can be employed to enhance the estimation of the \texttt{id}.\\
	
	This paper is organized as follows. Section \ref{sec:modelingbackground} presents the theoretical framework developed by \citet{Facco} from a statistical point of view. In Section \ref{sec:extension}, we contribute to the Poisson point process theory providing closed-form distributions for functions of distances between a point and its NNs. We exploit these novel results to devise estimators for the $\id$ of a dataset.
	Section \ref{sec::appl} presents numerical experiments devised to illustrate the behavior of the different estimators. Finally, in Section \ref{Sec::Concl} we discuss possible future directions and conclude.
	
	\section{The \texttt{TWO-NN} modeling background, revisited}
	\label{sec:modelingbackground}
	
	The two nearest neighbors (\texttt{TWO-NN}) model, proposed by \citet{Facco}, represents the foundation upon which we will build our contributions. 
	First, we discuss the theoretical background needed to derive the \texttt{TWO-NN} model. 
	Along with our exposition, we provide alternative, immediate proofs to the main theoretical results by exploiting the properties of known random variable distributions.\\
	
	Consider a dataset $\bm{X}=\{\bm{x}_i\}_{i=1}^n$ composed of $n$ observations measured over $D$ distinct features, i.e., $\bm{x}_i\in \mathbb{R}^D$, for $i=1,\ldots,n$. Denote with $\Delta:\mathbb{R}^D\times \mathbb{R}^D \rightarrow\mathbb{R}^+$ a generic distance function between the elements of $\mathbb{R}^D$.
	We assume that the dataset $\bm{X}$ is a particular realization of a Poisson point process characterized by density function (i.e., normalized intensity function) $\rho\left(\bm{x}\right)$. 
	We also suppose that the density of the considered stochastic process has its support on a manifold of unknown intrinsic dimension $d\leq D$. We expect, generally, that $d <<D$. \\
	For any fixed point $\bm{x}_i$, we can sort the remaining $n-1$ observations according to their distance from $\bm{x}_i$ by increasing order. Let us denote with $\bm{x}_{(i,l)}$ the $l$-th NN of $\bm{x}_i$ and with $r_{i,l}=\Delta(\bm{x}_{i},\bm{x}_{(i,l)})$ their distance, with $l=1,\ldots, n-1$. For practical purposes, we define $\bm{x}_{i,0}\equiv\bm{x}_{i}$ and $r_{i,0}=0$.\\
	A crucial quantity in this context is the \emph{volume of the hyper-spherical shell enclosed between two successive neighbors of} $\bm{x}_i$, defined as
	\begin{equation}
		v_{i,l}=\omega_{d}\left(r_{i,l}^{d}-r_{i,l-1}^{d}\right), \quad \quad \text{for  }l =1,\ldots,n-1,\text{ and  }i=1,\dots,n, \label{eq::HSshell}
	\end{equation}
	where $d$ is the dimensionality of the space in which the points are embedded (the \texttt{id}) and $\omega_{d}$ is the volume of the $d$-dimensional sphere with unitary radius. 
	Figure \ref{spheres} provides a visual representation of the introduced quantities in a three-dimensional case.\\ 
	It is worth noticing that in the univariate case each $v_{i,l}$ simplifies into the distance $\Delta(\bm{x}_{i},\bm{x}_{(i,l)})$ and it is called \emph{inter-arrival time}.
	If the underlying Poisson point process is \emph{homogeneous}, implying that $\rho(\bm{x})=\rho \:\: \forall \bm{x}$, all the $v_{i,l}$'s are independent and 
	identically distributed as an Exponential random variable, with rate parameter equal to the density $\rho$ \citep{Kingman1992}. 
	Building on the work of \citet{Moltchanov2012}, \citet{Facco} have extended this result to the multivariate case, where hyper-spherical shells defined as in \eqref{eq::HSshell} are the proper multivariate extension of the univariate inter-arrival times. Therefore, as in the univariate case, we have $v_{i,l}\sim Exp(\rho)$, for $l =1,\ldots,n-1,$ and $i=1,\dots,n$.\\
	
	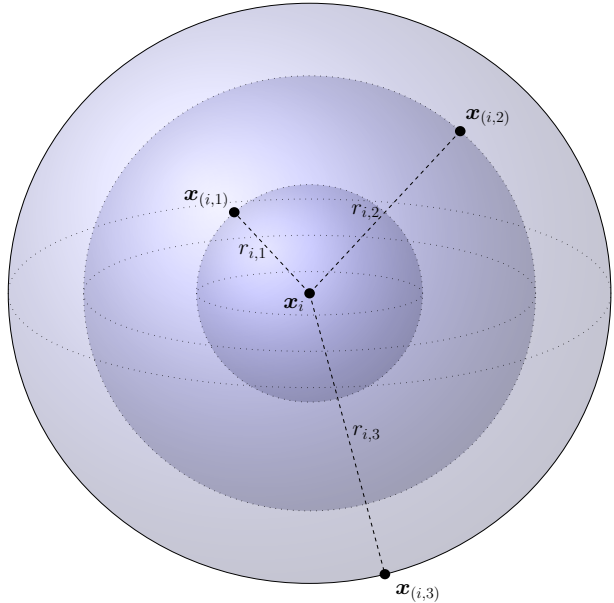
\begin{figure}[t!]
		\begin{center}
			\resizebox{8cm}{8cm}{%
				\begin{tikzpicture}
					\tikzstyle{every node}=[font=\LARGE]
					
					\shade[ball color = blue!40, opacity = 0.3] (0,0) circle (8cm);
					\shade[ball color = blue!40, opacity = 0.3] (0,0) circle (6cm);
					\shade[ball color = blue!40, opacity = 0.3] (0,0) circle (3cm);
					\draw (0,0) circle (8cm) node[below left] {$\bm{x}_i$};
					\draw[loosely dotted] (0,0) circle (6cm);
					\draw[loosely dotted] (0,0) circle (3cm);
					\draw[loosely dotted] (-8,0) arc (180:360:8 and 2.6);
					\draw[loosely dotted] (-6,0) arc (180:360:6 and 1.6);
					\draw[loosely dotted] (-3,0) arc (180:360:3 and 0.6);
					\draw[loosely dotted] (8,0) arc (0:180:8 and 2.6);
					\draw[loosely dotted] (6,0) arc (0:180:6 and 1.6);
					\draw[loosely dotted] (3,0) arc (0:180:3 and 0.6);

					\fill[fill=black] (0,0) circle (4pt);
					\fill[fill=black] (-2, 2.2360679775) circle (4pt) node[above left] {$\bm{x}_{(i,1)}$};
					\fill[fill=black] (4, 4.472135955) circle (4pt) node[above right] {$\bm{x}_{(i,2)}$};
					\fill[fill=black] (2, -7.74596669241) circle (4pt) node[below right = .2cm] {$\bm{x}_{(i,3)}$};
					
					\draw[dashed] (0,0 ) -- node[above,left]{$r_{i,1}$} (-2, 2.2360679775)  
					;
					\draw[dashed] (0,0 ) -- node[above,left]{$r_{i,2}$} (4, 4.472135955);
					\draw[dashed] (0,0 ) -- node[above,right]{$r_{i,3}$} (2, -7.74596669241);
				\end{tikzpicture}
			}
		\end{center}
		\caption{A pictorial representation in $\mathbb{R}^3$ of the quantities involved in the 
			\texttt{TWO-NN} modeling framework. The dots represent the data points. The selected observation, $\bm{x}_i$ is connected by dashed lines representing the distances $r_{i,j}$, $j=1,2,3$ to its first three NNs. The different spherical shells, characterized by different colors, have volume $v_{i,j}$, $j=1,2,3$. }  
		\label{spheres}
	\end{figure}
	Given these premises, the following theorem holds.
	
	\begin{theorem}
		Consider a distance function $\Delta$ taking values in $\mathbb{R}^+$ defined among the data points $\{\bm{x}_i\}_{i=1}^n$, which are a realization of a Poisson point process with constant density $\rho$. Let $r_{i,l}$ be the value of this distance between observation $i$ and its $l$-th NN. Then
		\begin{equation}\label{MOD1} 
			\mu_i = \dfrac{r_{i,2}}{r_{i,1}} \sim Pareto(1,d), \quad \quad \mu_i \in \left(1,+\infty\right).
		\end{equation}
		\label{Theo1}
	\end{theorem}
	In other words, using only basic properties of the homogeneous Poisson point process, \citet{Facco} showed that the ratio of the distances between a point and, respectively, its second and first NNs is Pareto distributed, with scale parameter equal to 1 and shape parameter $d$. Within this modeling framework, the latter parameter corresponds to the $\id$ of the dataset. Recall that if $Y \sim Pareto(a,b)$ then the density function of $Y$ is defined as $f_Y(y)=ab^a y^{-a-1}$, with $y\in (b,+\infty)$.
	The most important implication of Theorem \ref{Theo1} is that, once a proper distance is computed between the observations, we can summarize all the information 
	contained in the data
	about the $\id$ with the summary statistics given by $\bm{\mu}=\{\mu_i\}_{i=1}^n$, regardless the number of features  $D$ present in a dataset $\bm{X}$. This reduces the task of $\id$ estimation into a simple, scalable, and univariate estimation problem.
	A detailed proof of Theorem \ref{Theo1} is contained in \citet{tesiFacco}. Here, we provide an equivalent proof based on two properties of the Pareto distribution, that we now state. First, we remind that that $(\ast)$ if $X \sim Exp\left(\rho\right)$ and $Y \sim Erlang (n,\rho)$ such that $X \perp \!\!\! \perp  Y$, then $Z = \frac{X}{Y} + 1 \sim Pareto\left(1,n\right)$. Moreover, we can prove the following Lemma.
	\begin{lemma}[Scaling property of the Pareto distribution]
		$X\sim Pareto(1,\alpha) \iff Y=X^{q} \sim Pareto(1,\alpha/q).$
		\label{scaling}
	\end{lemma}
	\begin{proof}
		If $X\sim Pareto(1,\alpha)$, then $f_X(x) = \alpha x^{-(1+\alpha)}$.
		We consider the transformation $X=Y^{1/q}$ and compute $\frac{d}{dy}y^{1/q}=\frac{1}{q}y^{1/q-1}$. Then the density of $Y$ can be expressed as:
		\[ f_Y(y) = \alpha y^{-(1/q+\alpha/q)}\frac{1}{q}y^{1/q-1}= \left(\frac{\alpha}{q}\right) y^{(-\alpha/q +1)} \]
		which is the density of a $Pareto(1,\alpha/q)$ random variable. The converse can be shown by simply applying the inverse transformation.
	\end{proof}
	
	We are now ready to prove Theorem \ref{Theo1}.
	
	\begin{proof}
		Let us consider a generic point $\bm{x}_i$ and its corresponding volumes $\{v_{i,l}\}_{l=1}^{n-1}$ as defined in \eqref{eq::HSshell}. If the density of the Poisson point process is constant, 
		then $v_{i,l} \stackrel{i.i.d.}{\sim} Exp(\rho)$ for all $l$. Recall that the Exponential distribution is equivalent to an $Erlang(1,\rho)$ distribution. Then, according to $(\ast)$, $\frac{v_{i,2}}{v_{i,1}}+1 \sim Pareto(1,1)$. Also, we have that $\frac{v_{i,2}}{v_{i,1}}+1 = r_{i,2}^d/r^d_{i,1}$. We can then conclude that
		\[  \mu_i = \frac{r_{i,2}}{r_{i,1}} = \left( \frac{v_{i,2}}{v_{i,1}}+1 \right)^{1/d}, \]
		which according to Lemma \ref{scaling} implies $\mu_i \sim Pareto(1,d)$.
	\end{proof}
	
	We remark that, while the theorem can be proven only if the density is constant, the result and the $\id$ estimator are empirically valid as long as the density is approximately constant on the scale defined by the distance of the second NN $r_{i,2}$. We refer to this weakened assumption as \emph{local homogeneity}.\\
	
	The \texttt{TWO-NN} estimator treats the ratios $\mu_i$'s as independent, $i=1,\ldots,n$, and estimates the overall $\id$ $d$ on the entire dataset employing a least-squared approach. In detail, \citet{Facco} propose to consider the c.d.f. of each ratio $\mu_i$, given by $F({\mu_i})= (1-\mu_i^{-d})$, and to linearize it into $\log(1-F({\mu_i}))=-d\log(\mu_i)$. Then, a linear regression with no intercept is fitted to the pairs $\{-\log(1-\tilde{F}(\mu_{(i)})),\log(\mu_{(i)}) \}_{i=1}^n$, where $\tilde{F}(\mu_{(i)})$ denotes the empirical c.d.f. of the sample $\bm{\mu}$ sorted by increasing order. To enhance the estimation, the authors also suggested discarding the last percentiles 
	of the ratios $\mu_i$'s, usually generated by observations that fail to comply with the local homogeneity assumption.\\
	
	
	Before introducing other possible estimators for $d$, it is worth discussing the validity of the hypotheses we have made so far.
	As previously remarked, from a practical perspective we require that the density of the Poisson point process generating the data has to be locally constant, at least on the scale of the second NN of each point.
	In real applications, this hypothesis is satisfied if the available sample size is large enough, implying a densely populated space. 
	However, this assumption may fail in regions of the support where the data points are scarce. \\
	This issue is also linked to the \emph{curse of dimensionality} (CoD). 
	The sample size needed to produce a configuration of points that uniformly populates the manifold into consideration needs to scale exponentially with its dimension. To see how the CoD can affect the estimation, consider the theoretical setting in which we deal with a homogeneous Poisson process. The effect of the CoD becomes evident if we focus on the expected value and the variance of the random variable $\mu_i$, given by $\mathbb{E}\left[\mu_i\right]=d/(d-1)$ and $\mathbb{V}\left[\mu_i\right]=d/((d-1)^2(d-2))$, respectively. 
	If $d \rightarrow + \infty$, then the Pareto distribution collapses to a point mass in 1. Intuitively, when the dimensionality of the space that embeds the sample diverges, the distance between data points grows. As both the numerators and the denominators of all the elements in $\bm{\mu}$ scale with the same speed, they are asymptotically indistinguishable. 
	Therefore, in large dimensions, the hypothesis of local homogeneity is more likely to be violated when dealing with a fixed sample size. In those cases, the estimates based on the previous results are to be considered as lower bounds of the  true $d$ \citep{Ansuini2019}.\\
	
	The assumption of \emph{independence} among the elements of $\bm{\mu}$ allows the derivation of simple estimators for the parameter of interest. However, this is not always satisfied in practice because multiple observations can share the same NNs, and therefore the same distances. A possible solution would be to decimate the sample and eliminate the 
	NNs shared by multiple points before the analysis. However, as already shown in \citet{Allegra}, the estimates using the decimated samples do not substantially deviate from the ones obtained using all the data.\\
	
	Finally, the \texttt{TWO-NN} model does not directly consider the presence of noise in the dataset. Measurement errors can significantly impact the estimates since the $\id$ estimators are sensitive to scale effects. To exemplify, consider a dataset of 5,000 observations measured in $\mathbb{R}^3$ created as follows. The first two coordinates are obtained 
	from the spiral 
	defined by the parametric equations $x=u \cos(u)$ and $y= u\sin(u)$, where $u$ is sampled from a Uniform random variable with support $\left[\frac{1}{4\pi},1\right]$. The third coordinate is defined as a function of the previous two, $z = x^2 + y^2$. Gaussian random noise is added to all the three coordinates.
	A three-dimensional depiction of the resulting dataset is reported in the left half of Figure \ref{spiral3d}.
	The value of the $\id$ estimated with the \texttt{TWO-NN} model is 2.99. However, $u$ is the only stochastic quantity involved: all the coordinates are deterministically derived. Therefore, there is only one degree of freedom used in the data generating process. In other words, the true $\id$ is 1, and the noise misleads the \texttt{TWO-NN} estimator. \\
	In the next section, we will introduce novel estimators based on ratios of NNs distances of order higher than the second. By extending the order of NNs distances that we consider, we create estimators that can escape the short, ``local reach'' of the \texttt{TWO-NN} model, which is extremely sensitive to noise. Extending the neighborhood of a point to more NNs allows to extract meaningful information about the topology and the scale of the dataset at hand.\\
	
	\begin{figure}[ht]
		\begin{center}
			\includegraphics[scale=.25,trim=100 0 100 0,clip]{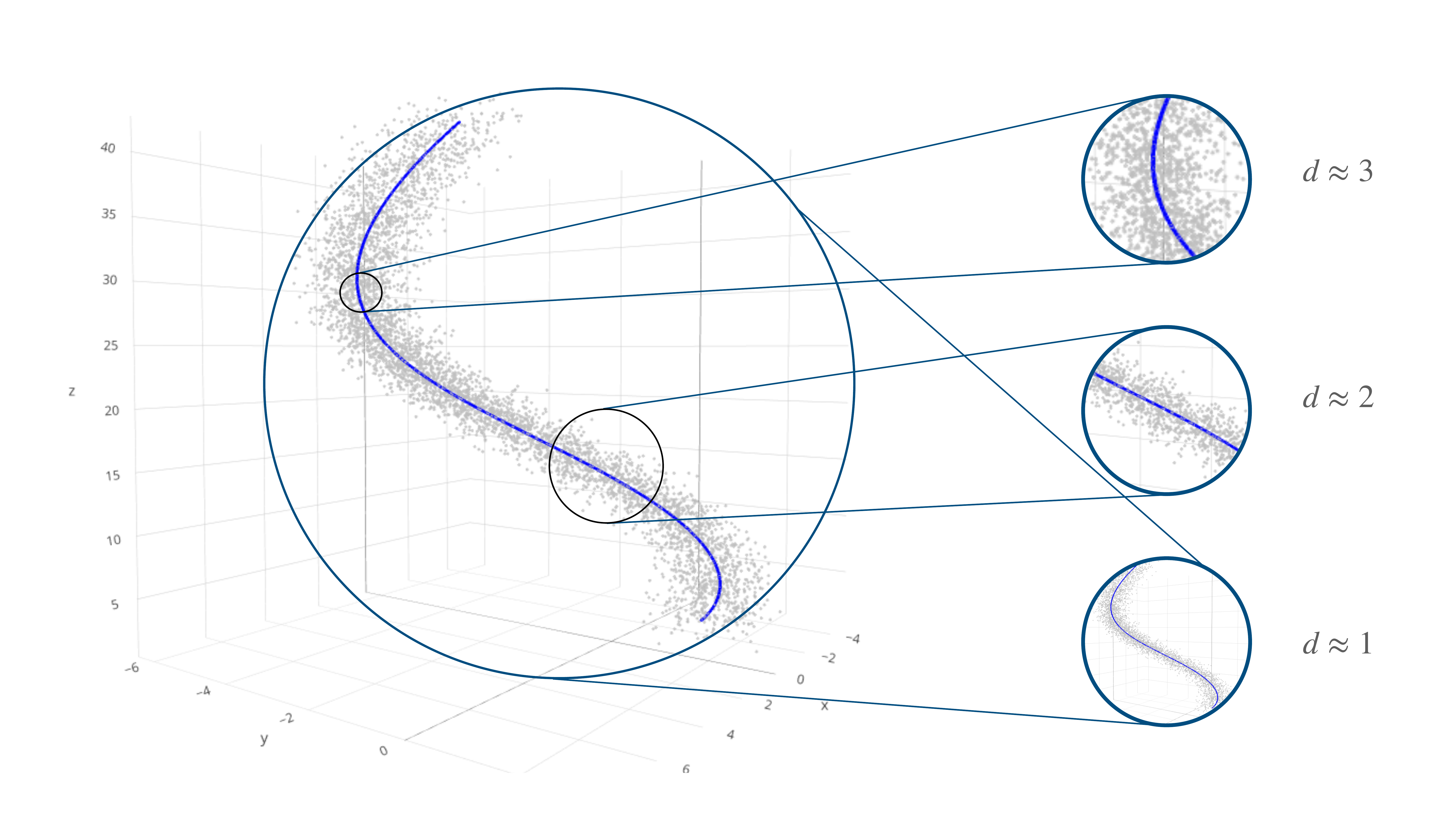}
		\end{center}
		\label{spiral3d}
		\caption{Three dimensional spiral. A dataset of 5,000 points is generated with deterministic transfomations starting from a Uniform sample. The resulting data points are displayed on the left. On the right, we show how 
			observing the data from different scales can 
			results in	different insights regarding the dimensionality of the data.}
	\end{figure}
	
	Since it is based on a simple linear regression, the \texttt{TWO-NN} estimator provides fast and accurate estimation of the $\id$, even when the sample size is large. Nonetheless, from \eqref{MOD1} we can immediately derive the corresponding Maximum Likelihood estimator (MLE) and the posterior distribution of $d$ under a Bayesian setting.\\
	Let us first discuss the MLE and the relative confidence intervals (CI). Trivially, for the shape parameter of a Pareto distribution the (unbiased) MLE is given by:
	
	\begin{equation}
		\hat{d} = \frac{n-1}{\sum_{i}^n\log(\mu_i)}.
		\label{MLE:TWONN}
	\end{equation}
	Moreover, $\hat{d}/d \sim IG(n,(n-1))$, where $IG$ denotes an Inverse-Gamma distribution. Therefore, the corresponding CI of level (1-$\alpha$) is given by  
	\begin{equation}
		CI(d,1-\alpha)=\left[\frac{\hat{d}}{q^{1-\alpha/2}_{IG_{n,(n-1)}}};\frac{\hat{d}}{q^{\alpha/2}_{IG_{n,(n-1)}}}\right],
		\label{CI}
	\end{equation}
	where $q^{\alpha/2}_{IG}$ denotes the quantile of order $\alpha/2$ of an Inverse-Gamma distribution.\\ 
	To carry out inference under the Bayesian approach, we specify a prior distribution on the parameter $d$. The most straightforward prior to choose is $d\sim Gamma(a,b)$ because of its conjugacy property. In this case, it is immediate to derive the posterior distribution:
	\begin{equation}
		d|\bm{\mu} \sim Gamma\left(a+n, b+\sum_{i=1}^n \log(\mu_i)\right).
	\end{equation}
	To perform model checking and assess the goodness of fit of the model in different scenarios, we can also compute the posterior predictive distribution.
	Let us define $a^*=a+n$ and $b^*=b+\sum_{i=1}^n \log(\mu_i)$. We obtain:
	\begin{equation}\label{Lomax}
		p(\tilde{\mu}|\bm{\mu})= \frac{a^*}{b^*\;\tilde{\mu}}\left(1+\frac{\log(\tilde{\mu})}{b^*}\right)^{-a^*-1}, \text{  with  } \: \tilde{\mu} \in \left(1,+\infty\right).
	\end{equation}
	From Equation \eqref{Lomax}, it can be easily shown that posterior predictive law for $\log(\tilde{\mu})$ follows a $Lomax(a^{*},b^{*})$ distribution, for which samplers are readily available.
	
	\section{Likelihood-based \texttt{id} estimators}
	\label{sec:extension}
	In this section, we develop novel theoretical results that contribute to Poisson point processes theory and that we will use to devise more precise estimators of $d$.
	In detail, we first extend the distributional results of Section \ref{sec:modelingbackground} providing closed-form distributions for vectors of consecutive ratios of distances and ratios of NNs of generic order. In both cases, we derive the corresponding estimators for the $\id$ parameter.
	
	\subsection{Distribution of consecutive ratios and the \texttt{Cride} estimator}
	Consider the same setting introduced in the previous section and  
	define $V_{i,l} = \omega_d \, r^d_{i,l}$ as the volume of the hyper-sphere centered in $\bm{x}_i$ with radius equal to the distance between $\bm{x}_i$ and its $l$-th NN. Because of their definitions, for $l=2,\ldots,L$, we have that $v_{i,l}$ and $V_{i,l-1}=v_{i,1}+\cdots +v_{i,l-1}$ are independent. Moreover, $V_{i,l}\sim Erlang(1,l-1)$. Then, we can write
	\begin{equation}
		\frac{v_{i,l}}{V_{i,l-1}} = \frac{\omega_d \left(r_{i,l}^d-r_{i,l-1}^d  \right)}{\omega_d r^d_{i,l-1}} =\left( \frac{r_{i,l} }{ r_{i,l-1}}\right)^d-1,
	\end{equation}
	which becomes, after a little algebra,
	\begin{equation}
		\label{mu1}
		\mu_{i,l}=\frac{r_{i,l} }{ r_{i,l-1}} = \left(\frac{v_{i,l}}{V_{i,l-1}}+1\right)^{1/d}.
	\end{equation}
	Given these premises, the following theorem holds.

	\begin{theorem}
		Consider a distance $\Delta$ taking values in $\mathbb{R}^+$ defined among the data points $\{\bm{x}_i\}_{i=1}^n$, which are a realization of a Poisson point process with constant density $\rho$. Let $r_{i,l}$ be the value of the distance between observation $i$ and its $l$-th NN. Define $\mu_{i,l} =r_{i,l} / r_{i,l-1}$.
		It follows that
		\begin{equation}
			\begin{aligned}
				\mu_{i,l}& \sim  Pareto(1,(l-1)d), \quad \text{  for  }\quad  l = 2,\ldots,L.
			\end{aligned}
			\label{musgammas}
		\end{equation}
		Moreover, the elements of the vector $\bm{\mu_{i,L}}=\{ \mu_{i,l} \}_{l=2}^L$ are jointly independent.
		\label{Theo2}
	\end{theorem}
	
	\begin{proof}
		The marginal distributions stated in Equation \eqref{musgammas}  follow by the application of elementary properties of Exponential, Gamma and Pareto random variables to Equation \eqref{mu1}. \\
		We now prove that the joint independence of the elements of the vector $\bm{\mu_{i,L}}$. We drop the observational index $i$ for ease of exposition. Let us denote $\gamma_l=\log\left(\frac{r_{l}}{r_{l-1}}\right)$, for $l=2,3,\ldots,L$. We want to derive the joint density of $\bm{\gamma}_L=\left(\gamma_2,\ldots,\gamma_L\right)$.
		To do so, we start from the joint density of $\left(v_1,v_2,v_3,\ldots,v_L\right)$, denoted by $f(v_1,v_2,v_3,\ldots,v_L) = \rho^L \exp\left[-\rho\sum_{l=1}^{L}v_l\right].$
		Consider the following one-to-one transformation of the vector $\bm{\gamma}_L$:
		$$  \gamma_1 = v_1  \quad \text{  and  }\quad \gamma_l = \frac{1}{d} \log\left(1+\frac{v_l}{\sum_{k=1}^{l-1}v_k}\right), \:\: l=2,\ldots,L $$
		with inverse
		$$  v_1 = \gamma_1  \quad \text{  and  }\quad v_l = \gamma_1\exp\left(d\sum_{k=2}^{l-1}\gamma_k\right) \left(\exp\left( d\gamma_l\right)-1\right), \:\: l=2,\ldots,L  .$$
		The determinant of the Jacobian matrix $J$ associated with this transformation is \\$|J| = \gamma^{L-1}_1 d^{L-1} \prod_{l=2}^{L} \exp\left[d\cdot(L-l+1)  \gamma_l\right].$ \\
		Consequently, the density of the transformed vector is 	
		\begin{align*}
			f(\bm{\gamma}) &=\rho^L\gamma^{L-1}_1 d^{L-1}\: \exp\left[-\rho \gamma_1 \exp\left(d\sum_{l=2}^{L}\gamma_l \right)\right] \prod_{l=2}^{L}\: \exp\left[d\cdot(L-l+1)  \gamma_l\right].
		\end{align*} 
		We then integrate out $\gamma_1$ to obtain:
		\begin{align*}
			f\left(\bm{\gamma}_L\right) =& d^{L-1} \prod_{l=2}^{L}(l-1) \exp\left[-(l-1)d\gamma_l\right] =  \prod_{l=2}^{L} (l-1)d \exp\left[-(l-1)d\gamma_l\right].
		\end{align*}
		Since $f\left(\bm{\gamma}_L\right)=\prod_{l=2}^{L}f\left(\gamma_l\right)$, we can conclude that $\gamma_2,\ldots,\gamma_L$ are independent exponential random variables. Finally, given that $X\sim Pareto(1,a)\iff \log(X) \sim Exp(a)$, we  consider $\bm{\mu}_L=\exp(\bm{\gamma}_L)$ and conclude the proof.
	\end{proof}

	Theorem \ref{Theo2} provides a way to characterize the distributions of consecutive ratios of distances. Remarkably, given the homogeneity assumption, the different ratios are all independent. Therefore, since all of the $L-1$ densities depend on the same shape parameter $d$, we can derive an estimator that can use more information extracted from the data. The (unbiased) MLE in this case becomes
	
	\begin{equation} \label{MLECride}
		\hat{d}_L = \frac{n(L-1)-1}{\sum_{i=1}^n\sum_{l=2}^L (l-1)\log(\mu_{i,l})}.
	\end{equation}
	
	This estimator has variance $\mathbb{V}\left[\hat{d}_L\right]=d^2/(n(L-1)-2)$ which is smaller that the variance of the MLE estimator in \eqref{MLE:TWONN}, that is recovered when $L=2$. The CI is analogous to \eqref{CI}, with $n$ substituted by $n(L-1)$. From a Bayesian perspective, we can again specify a conjugate Gamma prior for $d$, obtaining the posterior distribution
	
	\begin{equation}
		\hat{d}_L|\bm{\mu}_{L} \sim Gamma \left( a+n(L-1),b+\sum_{i=1}^n\sum_{l=2}^L (l-1)\log(\mu_{i,l})\right).
	\end{equation}
	
	Alternatively, one can go back to the univariate modeling case by considering the transformation $\gamma_{i,l}=\log\left(\mu_{i,l} \right)$, obtaining that $\bm{\gamma}_{i,l} \sim  Exp((l-1)d)$ and define
	\begin{equation}
		\Gamma_{i,L}=\sum_{l=2}^{L} (l-1) \cdot \gamma_{i,{l}} \sim Erlang\left( L-1 ,d\right), \quad i=1,\ldots,n.
		\label{alternatives1}
	\end{equation}
	It can be proven that the MLE obtained from \eqref{alternatives1} is identical to the one presented in Equation \eqref{MLECride}. We name the estimators derived from Theorem \ref{Theo2} the Consecutive Ratios $\id$ Estimators -- \texttt{Cride}.
	We remark that many other distributions can be employed using the properties of the Exponential random variables. As an example, for a generic observation $i$ and a generic ratio of order $l$, the following statements are equivalent to \eqref{alternatives1}:
	\begin{equation*}
		\gamma_{i,l}^2 \sim Weibull\left(\frac12,\frac{1}{(l-1)^2d^2} \right),\quad \mu - \sigma\log\left((l-1)d\gamma_{i,l})\right) \sim GEV\left(\mu,\sigma,0\right), 
		\label{alternatives2}
	\end{equation*}
	where $GEV$ indicates the Generalized Extreme Values distribution \citep{McFadden1978}. These distribution are well known in Extreme Value Theory (EVT). Other authors have recently developed an \texttt{id} estimator in an EVT framework \citep{Amsaleg2015,Houle2013}: we leave the investigation of potential connections among these two fields for future research. \\
	
	To conclude this subsection, we underline that the structure of the MLE estimators \eqref{MLECride} (and consequently \eqref{MLE:TWONN}) is equivalent to the one proposed in \citet{Levina} when focusing on one single data point, since $\sum_{l=2}^L (l-1) \log(\mu_{i,l})$ can be rewritten $\sum_{l=1}^{L-1} \log(r_{L}/r_{l})$. This result is unsurprising: despite following different derivations, we started from the same premises, as already underlined in \citet{Facco}. However, the main difference is how the estimators combine the information extracted from the entire dataset. Our theoretical derivation naturally leads to average the inverses of the contributions to the likelihood of every single data point rather than considering a simple average. To this extent, we see that \texttt{Cride} is equivalent to the estimator proposed in a comment by \citet{Comment}. Although the resulting MLEs are the same, we believe that our approach presents an advantage. Indeed, starting from the distributions of the ratios of NNs distances, we can effortlessly derive uncertainty quantification estimates, as in \eqref{CI}, by simply exploiting well-known properties of the Pareto random variable. In the following subsection, we present another estimator that relies on a single ratio of distances for each data point (similarly to the \texttt{TWO-NN}) while considering information collected on larger neighbors (similarly to \texttt{Cride}).

	\subsection{Distributions of generic ratios, distances, and \texttt{Gride}}
	
	Building of the previous statements, we can derive more general results about the distances between NNs from a Poisson point process realization. The next theorem characterizes the distribution of the ratio of distances from two NNs of generic order. 
	
	\begin{theorem}
		Consider a distance $\Delta$ taking values in $\mathbb{R}^+$ defined among the data points $\{\bm{x}_i\}_{i=1}^n$, which are a realization of a Poisson point process with constant density $\rho$. Let $r_{i,l}$ be the value of this distance between observation $i$ and its $l$-th NN. Consider two integers $1\leq n_1<n_2$ 
		and define $\dot{\mu}=\mu_{i,n_1,n_2} = r_{i,n_2} / r_{i,n_1}$.
		The random variable $\dot{\mu}$ is characterized by density function 	
		\begin{equation}
			\label{mudot}
			f_{\mu_{i,n_1,n_2}}(\dot{\mu})= 
			\frac{d(\dot{\mu}^d-1)^{n_2-n_1-1}}{\dot{\mu}^{(n_2-1)d+1}B(n_2-n_1,n_1)}, \quad \dot{\mu}>1,
		\end{equation}
		where $B(\cdot,\cdot)$ denotes the Beta function. Moreover, $\dot{\mu}$ has $k$-th moment given by
		\begin{equation}
			\mathbb{E}\left[\dot{\mu}^k\right]=\frac{B(n_2-n_1,n_1-k/d)}{B(n_2-n_1,n_1)}.
		\end{equation} 
		\label{Theo3}
	\end{theorem}
	\begin{proof}
		Let $\{W_i\}_{i=1}^n$, $n \geq 2$, denote a sequence of independent Exponential random variables with pairwise distinct parameters $\lambda_i$. The sum of $n$ random variables $W_i\sim Exp(\lambda_i)$ is said to follow an hypo-exponential distribution, with density
		$$f_{\sum_{i=1}^{n}W_i}(w)=\left[\prod_{i=1}^{n} \lambda_{i}\right] \sum_{j=1}^{n} \frac{\mathrm{e}^{-\lambda_{j} w}}{\prod_{l \neq j \atop l=1}^{n}\left(\lambda_{l}-\lambda_{j}\right)}, \quad w>0. $$	
		Our goal is to characterize the distribution of $\dot{\mu}=\mu_{i,n_1,n_2} = \frac{r_{n_2}}{r_{n_1}}$, with $n_2>n_1$ integer values. 
		First, we notice that $\dot{\mu}$ can be rewritten as telescopic product of $n_2-n_1$ ratios, all independent and Pareto distributed:	
		\[\dot{\mu} = \frac{r_{n_2}}{r_{n_1}}=  \frac{r_{n_2}}{r_{n_2-1}}\cdot \frac{r_{n_2-1}}{r_{n_2-2}}\cdots \frac{r_{n_1+1}}{r_{n_1}}.\]
		Define $\gamma_l=\log\left(\frac{r_{l}}{r_{l-1}}\right)$ and consider $Y=\log(\dot{\mu})$. Then, we can write $Y = \log (\dot{\mu})= \log\left(\frac{r_{n_2}}{r_{n_1}}\right)=  \sum_{l=n_1+1}^{n_2}\gamma_{l}$.
		Since each $\gamma_l$ is defined as the logarithm of a Pareto distribution, we have just shown that $Y$ is a sum of $L=n_2-n_1$ independent Exponential random variables with parameters ranging from $n_1d$ to $(n_2-1)d$. Plugging these parameters into the the definition of hypo-exponential density, we can write the distribution of $Y$ as
		\begin{equation}
			\begin{aligned}
				f_{Y}(y)=
				d \frac{(n_2-1)!}{(n_1-1)!} \sum_{j=1}^{n_2-n_1} \frac{\mathrm{e}^{- (n_1+j-1)dy}}{\prod_{l \neq j \atop l=1}^{n_2-n_1}\left(l-j\right)}, \quad y>0.\\
			\end{aligned}
			\label{hypoexp}
		\end{equation}
		From here, we derive the distribution for $\dot{\mu}=\exp(Y)$, transforming the last density in \eqref{hypoexp}.
		\begin{equation*}
			\begin{aligned}
				f_{\dot{\mu}}(\dot{\mu})=& d \frac{(n_2-1)!}{(n_1-1)!} \frac{1}{\dot{\mu}} \sum_{j=1}^{n_2-n_1} \frac{\mathrm{e}^{- (n_1+j-1)d\log \dot{\mu}}}{\prod_{l \neq j \atop l=1}^{n_2-n_1}\left(l-j\right)} 
				= d \frac{(n_2-1)!}{(n_1-1)!} \sum_{j=1}^{n_2-n_1} \frac{\dot{\mu}^{- (n_1+j-1)d-1}}{\prod_{l \neq j \atop l=1}^{n_2-n_1}\left(l-j\right)},\\=&
				d \frac{(n_2-1)!}{(n_1-1)!} \sum_{j=1}^{n_2-n_1} \frac{\dot{\mu}^{- (n_1+j-1)d-1}}{(j-1)!(n_2-n_1-j)!(-1)^{j-1}}\\
				=& d \frac{(n_2-1)!}{(n_1-1)!} \sum_{k=1}^{n_2-n_1} \frac{\dot{\mu}^{- (n_2-k)d-1}}{(k-1)!(n_2-n_1-k)!(-1)^{n_2-n_1-k}}\\=&
				\frac{d}{\dot{\mu}^{n_2d+1}} \frac{(n_2-1)!}{(n_1-1)!} \sum_{k=1}^{n_2-n_1} \frac{\dot{\mu}^{kd}(-1)^{n_2-n_1-k}}{(k-1)!(n_2-n_1-k)!}\\=&
				\frac{d}{\dot{\mu}^{n_2d+1}} \frac{(n_2-1)!}{(n_1-1)!} 
				\frac{(n_2-n_1-1)!}{(n_2-n_1-1)!}
				\sum_{k=1}^{n_2-n_1} \frac{\dot{\mu}^{kd}(-1)^{n_2-n_1-k}}{(k-1)!(n_2-n_1-k)!}\\ 
				=&
				\frac{d}{\dot{\mu}^{n_2d+1}} \frac{(n_2-1)!}{(n_1-1)!} 
				\frac{1}{(n_2-n_1-1)!}
				\sum_{l=0}^{n_2-n_1-1}
				\binom{n_2-n_1-1}{l} {(\dot{\mu}^{d})}^{l+1} (-1)^{n_2-n_1-l-1}\\
				=&
				\frac{d}{\dot{\mu}^{(n_2-1)d+1}} \frac{(n_2-1)!}{(n_1-1)!} 
				\frac{(\dot{\mu}^d-1)^{n_2-n_1-1}}{(n_2-n_1-1)!}
				= (n_2-n_1) 
				\binom{n_2-1}{n_1-1}
				\frac{d(\dot{\mu}^d-1)^{n_2-n_1-1}}{\dot{\mu}^{(n_2-1)d+1}}
				\\
				=& 
				\frac{d(\dot{\mu}^d-1)^{n_2-n_1-1}\cdot \dot{\mu}^{-(n_2-1)d-1} }{B(n_2-n_1,n_1)}, \quad \dot{\mu}>1.
			\end{aligned}
		\end{equation*}
		In the previous derivation, we applied the following equality at the second line: $\prod_{l \neq j \atop l=1}^{n_2-n_1}\left(l-j\right)= (j-1)!(n_2-n_1-j)!(-1)^{j+1}$. Moreover, at the fourth line we applied the reflection property of the indexes of a sum: $\sum_{k=1}^{K}a_k=\sum_{k=1}^{K} a_{K-k+1}$. At the sixth line, we applied the Newton binomial formula.	
		Interestingly, we can define $Z=\dot{\mu}^d-1$ to find that 
		$Z\sim\beta'(n_2-n_1,n_1)$, where $\beta'$ denotes the \emph{Beta prime} distribution. This property helps to find the expression for the generic moment of $\dot{\mu}$:
		\begin{equation*}
			\mathbb{E}\left[\dot{\mu}^k\right] = \mathbb{E}\left[(Z+1)^{k/q}\right] = \int_{0}^{+\infty} (z+1)^{1/q} \frac{z^{n_2-n_1}(1+z)^{-n_2}}{B(n_2-n_1,n_1)}dz=\frac{B(n_2-n_1,n_1-k/d)}{B(n_2-n_1,n_1)},
		\end{equation*}
		that is well-defined for $k<dn_1.$
	\end{proof}

	\begin{remark}[1]
		Given the expression of the generic moment of $\dot{\mu}$, we can derive its expected value and variance:
		\begin{equation}
			\mathbb{E}\left[\dot{\mu}\right] = \frac{B(n_2-n_1,n_1-1/d)}{B(n_2-n_1,n_1)}\quad \text{and} \quad \mathbb{V}\left[\dot{\mu}\right]=\frac{B(n_2-n_1,n_1-2/d)}{B(n_2-n_1,n_1)}-\frac{B(n_2-n_1,n_1-1/d)^2}{B(n_2-n_1,n_1)^2},
		\end{equation}
		both well-defined when $d>2$.
	\end{remark}
	
	\begin{remark}[2]
		Formula \eqref{mudot} can be specialized to the case where $n_1=n_0$ and $n_2=2n_0$. We obtain
		\begin{equation} \label{2n}
			f_{\mu_{i,n_0,2n_0}}(\dot{\mu})= 
			\frac{(2n_0-1)!}{(n_0-1)!^2} \cdot
			\frac{d(\dot{\mu}^d-1)^{n_0-1}}{\dot{\mu}^{(2n_0-1)d+1}} = 
			\frac{d(\dot{\mu}^d-1)^{n_0-1}}{B(n_0,n_0)\cdot \dot{\mu}^{(2n_0-1)d+1}} , \quad \dot{\mu}>1.
		\end{equation}		
	\end{remark}
	
	\begin{remark}[3]
		The result at the basis of the \texttt{Cride} model in Equation \eqref{musgammas} can be derived as special case of formula \eqref{mudot}. Consequently, we can say the same for the \texttt{TWO-NN} model in Equation \eqref{MOD1}. Specifically, if we set $n_1=n_0$ and $n_2=n_0+1$, we obtain
		\begin{equation} \label{cride2}
			f_{\mu_{i,n_0,n_0+1}}(\dot{\mu})= 
			n_0 d \dot{\mu}^{-n_0d-1}, \quad \dot{\mu}>1,
		\end{equation}	
		which is the density of a $Pareto(1,n_0d)$ distribution.	
	\end{remark}
	
	The distributions reported in Equations \eqref{mudot} and \eqref{2n} allow us to devise a novel estimator for the $\id$ parameter, based on the properties of the distances measured between a point and two of its NNs of generic order. We name this method the Generalized Ratios $\id$ Estimator (\texttt{Gride}). We were not able 
	to derive a closed-form MLE in this case, but the estimation can be easily carried out employing one-dimensional numerical optimizing techniques. Moreover, numerical methods can be exploited for uncertainty estimation: for example, one can obtain the 
	estimated confidence intervals with parametric bootstrap.
	We display some examples of the shapes of the density functions defined in Equation \eqref{mudot} in Figure \ref{fig:mudots}.\\
	\begin{figure}[ht!]
		\begin{center}
			\includegraphics[scale=.5]{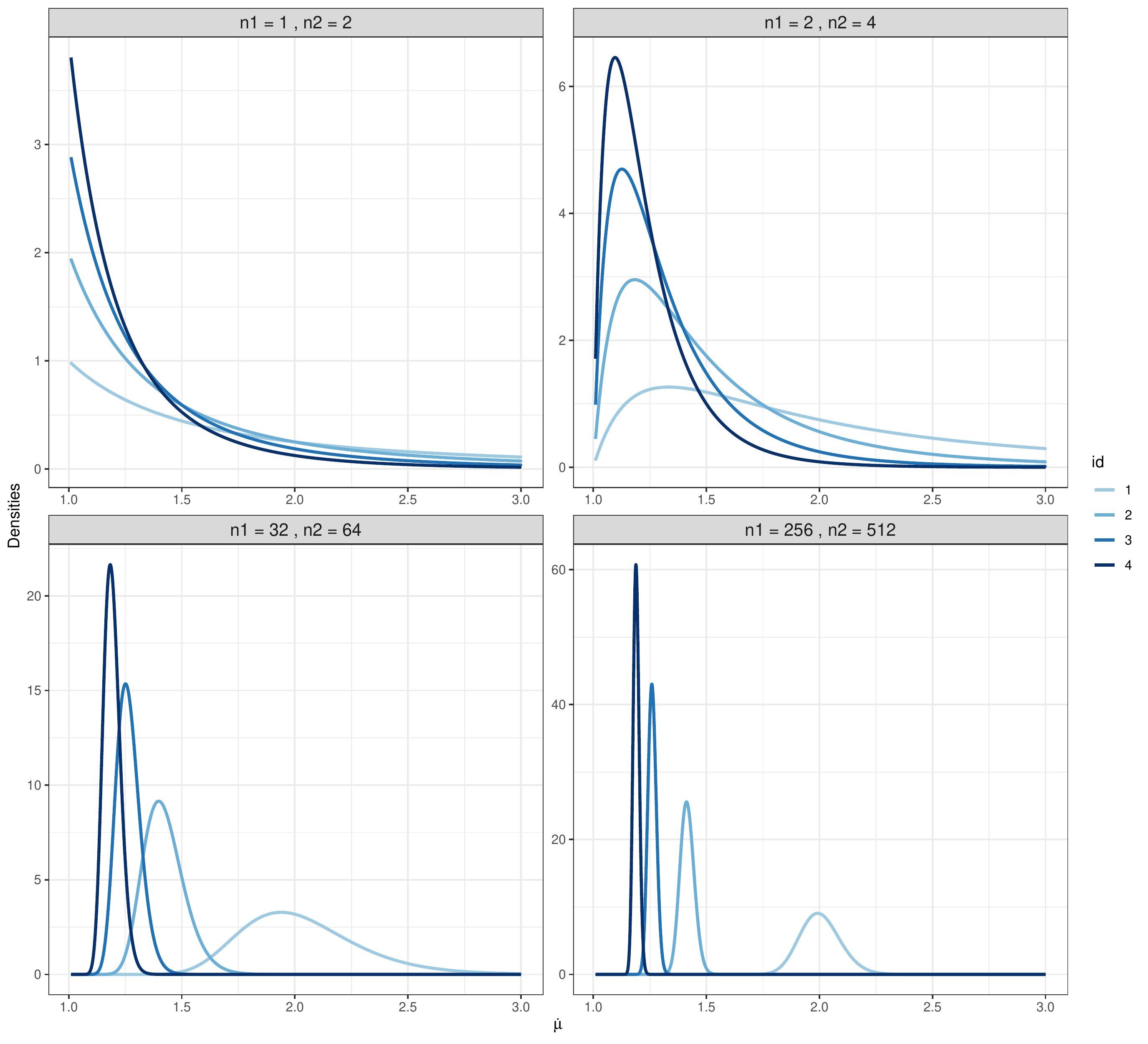}
		\end{center}
		\caption{Examples of density functions characterizing the random variable $\dot{\mu}$ as defined in Equation \eqref{mudot}. The different colors correspond to different values of the $\id$ parameter $d$, while the panels display the various order of NNs considered for the ratios.}
		\label{fig:mudots}
	\end{figure}

	Notice that, similarly to Theorem \ref{Theo1}, Theorems \ref{Theo2} and \ref{Theo3} can be proven only assuming $\rho$ to be constant. However, from a practical perspective, the novel estimators are empirically valid as long as the density $\rho$ is approximately constant on the scale defined by the distance of the $L$-th NN $r_{i,L}$ (\texttt{Cride}) and the $n_2$-th NN $r_{i,n_2}$ (\texttt{Gride}), respectively. Again, we will refer to this assumption as local homogeneity. We provide a more detailed discussion of these assumptions in Section \ref{assumpt}. In the next subsection, we discuss the advantages of an estimator built on the results stated in Theorem \ref{Theo3}.\\
	
	We conclude this section providing another theoretical results. We derive a closed-form expression for the joint density of the random distances between a point and its first $L$ NNs for a homogeneous Poisson point process. We defer the proof of the next theorem to the Appendix.
	
	\begin{theorem}
		Consider a distance $\Delta$ taking values in $\mathbb{R}^+$ defined among the data points $\{\bm{x}_i\}_{i=1}^n$, which are a realization of a Poisson point process with constant density $\rho$. Let $r_{i,l}$ be the value of this distance between observation $i$ and its $l$-th NN. 
		Then, the joint distribution of the the vector $\left(r_{i,1},\ldots,r_{i,L}\right)$ is given by
		\begin{equation}
			f(r_{i,1},\ldots,r_{i,L})= (\rho \omega_d d)^L  \left(\prod_{l=1}^{L}r_{i,l}^{d-1} \right)\exp\left[ -\rho\omega_d r_{i,L}^d \right],
		\end{equation}
		with $r_{i,l} \in \mathbb{R}^+$, and the constraint that $r_{i,1}<r_{i,2}<\ldots<r_{i,L}$. Moreover, the marginal random distance between a point $\bm{x}_i$ and its $L$-th NN has density
		\begin{equation}
			f(r_{i,L})=\exp\left[ -\rho\omega_d r_{i,L}^d \right] (\rho \omega_d d)^L \frac{r_{i,L}^{Ld-1}}{(L-1)!d^{L-1}}.
		\end{equation}
		This result implies that, for $i=1,\ldots,n$, 
		$r_{i,L}
		\sim GenGamma(p,a,q)$, i.e., it follows a Generalized Gamma distribution with parameters $p=d$, $a=1/\sqrt[d]{\rho \omega_d}$, and $q=Ld$.	
		\label{Lemma4}
	\end{theorem}

	\subsection{The assumptions behind \texttt{Cride} and \texttt{Gride}}
	\label{assumpt}
	The estimators presented in the previous subsections extend the \texttt{TWO-NN} rationale to broader neighborhoods. By considering a larger number of NNs, the models consider more information regarding the topology of the data configuration. As a consequence, ratios of higher NNs orders allow the investigation of the relationship between the dataset $\id$  and the width of the neighborhood. 
	That way, we can escape the strict, extremely local point of view of the \texttt{TWON-NN}, which allows us to reduce the distortion produced by noisy observations in the estimation of the $\id$.\\
	However, to understand when the results obtained with \texttt{Cride} and \texttt{Gride} are reliable in real settings, we need to discuss the validity of the assumptions needed for their derivations. As mentioned in Section \ref{sec:modelingbackground}, the main modeling assumptions are two: the local homogeneity of the density of the underlying Poisson point process and the independence among ratios of distances centered in different data points. To provide a visual comparison, we display in Figure \ref{fig::diff_models} an example. We consider 500 points generated from a bidimensional Uniform distribution over the unit square. Then, we select four points (in blue) and highlight (in red) the NNs involved in the computation of the ratios that are used by the \texttt{TWO-NN}, \texttt{Cride}, and \texttt{Gride} models. For \texttt{Cride}, we set $L=40$. For \texttt{Gride}, $n_1=20$ and $n_2=40$.
	\begin{figure}[ht!]
		\begin{center}
			\includegraphics[width=\linewidth]{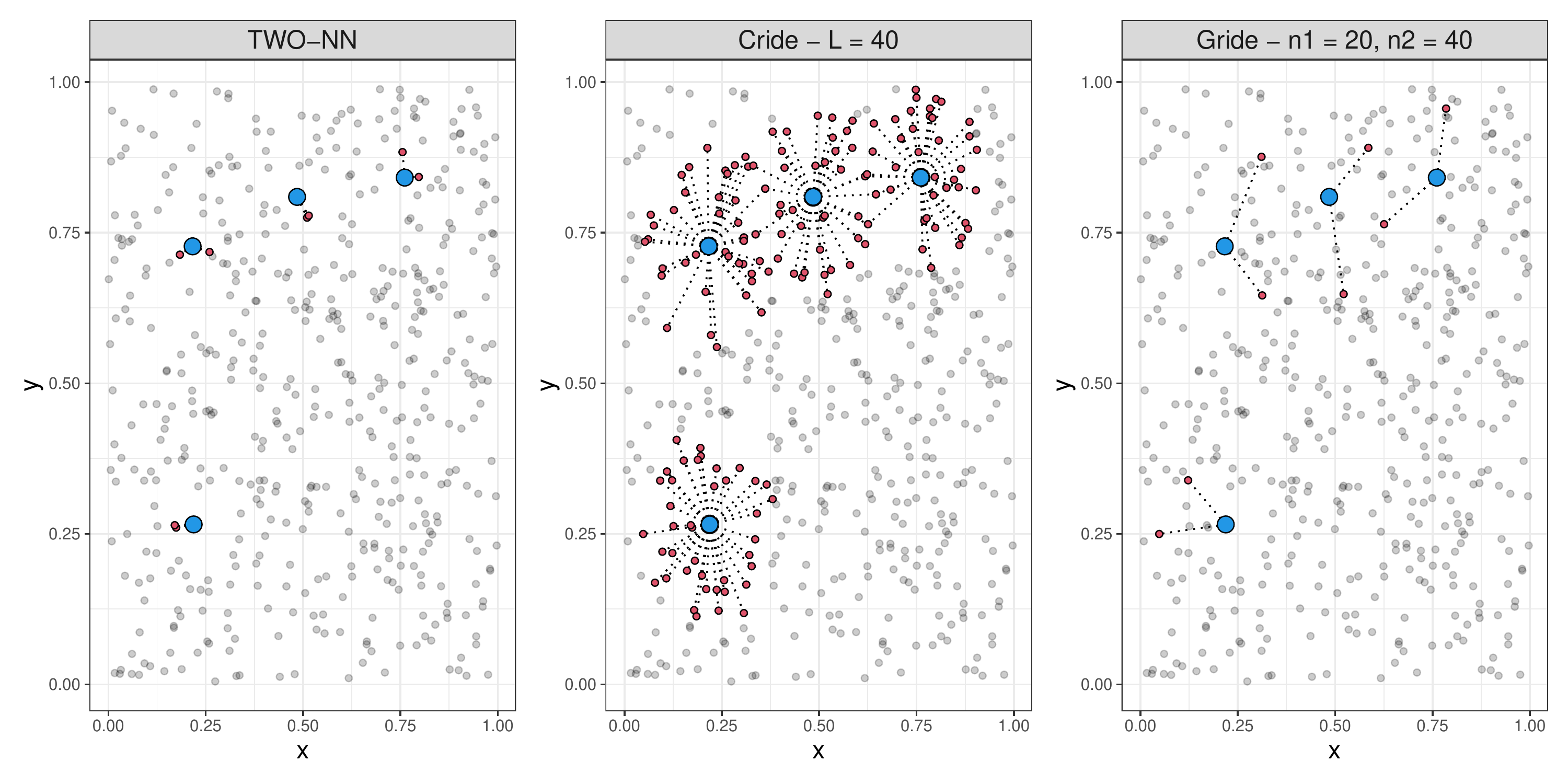}
		\end{center}	
		\caption{Neighboring points (in red) and distances (dotted lines) involved in the $\id$ estimation centered in four data points (in blue). Each panel corresponds to one model: \texttt{TWO-NN}, \texttt{Cride}, and \texttt{Gride}, respectively.}
		\label{fig::diff_models}
	\end{figure}
	
	First, we notice that in the \texttt{Cride} and \texttt{Gride} cases the local homogeneity hypothesis has to hold for larger neighborhoods, up to the NN of order $L>2$ and $n_2>2$, respectively. 
	As we will prove empirically, while the two novel estimators are more reliable than \texttt{TWO-NN} if used on dense configurations, when the data become scarce care should be used when interpreting the results. Although the stricter local homogeneity assumption affects the two novel estimators similarly, they are not equally impacted by the assumption of independence of the ratios. 
	By comparing the second and third panels of Figure \ref{fig::diff_models}, we observe that \texttt{Cride}, in its computation, needs to take into account all the distances between points and its NNs up to the $L$-th order. When $L$ is large and the sample size is limited, neighborhoods centered in different data points may overlap, inducing dependence across the ratios. \texttt{Gride} instead uses only two of the $L$ distances, and the probability of shared NNs across different data points is lower, especially if large $n_1$ and $n_2$ are chosen.

	\section{Numerical Experiments}
	\label{sec::appl}
	
	First, we empirically show that the variance of the \texttt{Gride} estimator is reduced as we consider NNs of higher order. This represents an important gain with respect to the \texttt{TWO-NN} estimator. We sample $10,000$ observations from a bivariate Gaussian distribution, and aim at estimating the true $\texttt{id}=2$. To assess the variance of the numerical estimator devised from Equation \eqref{cride2}, we resort to parametric bootstrap techniques. We collect $5,000$ simulations as bootstrap samples under four different scenarios that we report in the first row of Figure \ref{fig:bootstr}. A similar analysis can be performed within the Bayesian setting, studying the concentration of the posterior distribution. We display the posterior simulations on the second row of the same figure. We see that, as the NNs order increases, 
	the bootstrap samples (top row) and the posterior samples (bottom row)
	are progressively more concentrated around the truth, 
	with  minor remaining bias due to the lack of perfect homogeneity in the data generating process. \\
	\begin{figure}[t!]
		\begin{center}
			\includegraphics[scale=.4]{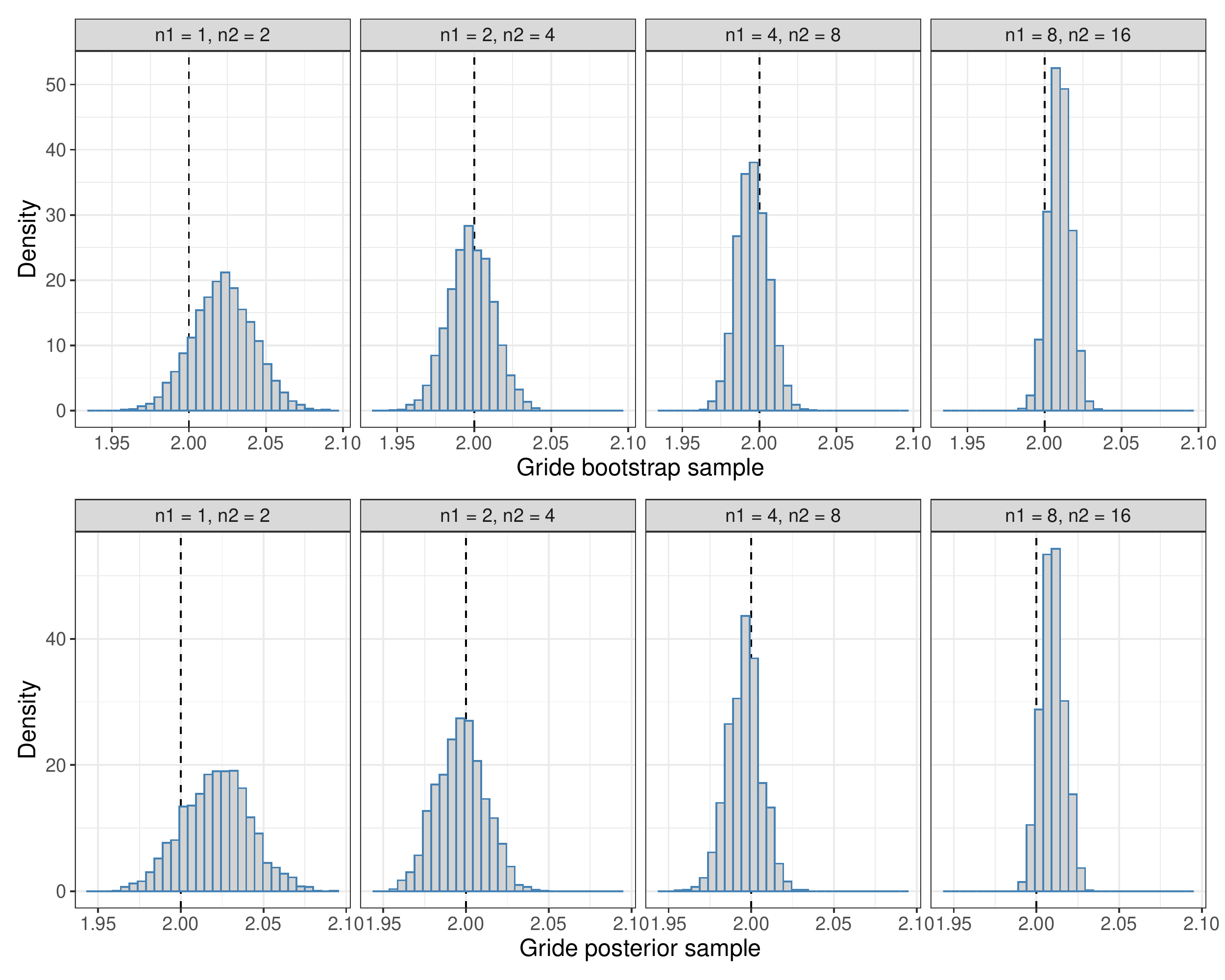}
		\end{center}
		\caption{Histograms of the parametric bootstrap samples (top row) and posterior samples (bottom row) for the \texttt{Gride} models estimated withing the frequentist and Bayesian framework. The panels in the first column correspond to the \texttt{TWO-NN} model.}
		\label{fig:bootstr}
	\end{figure}
	
	As a second analysis, we empirically show that high-order \texttt{Gride} estimates are also less biased than the ones obtained with the \texttt{TWO-NN} model when the homogeneity assumption of the underlying Poisson process holds. 
	In general, we cannot expect this assumption to be perfectly met in real datasets due to density variations, boundary effects, and noise in the observations. However, it is essential to develop estimators that can perform well in a reference ideal condition.\\
	To create a dataset that complies as much as possible with the theoretical data-generating mechanism, we start by fixing a pivot point and we generate a sequence of $N=40,000$ volumes of hyperspherical shells from an Exponential distribution, under the homogeneous Poisson process framework. Let us denote the sequence of these volumes with $\left\{v_j \right\}_{j=1}^{N}$. Once the volumes are collected, we  compute the actual distance (radius) from the pivot point by using Equation \eqref{eq::HSshell} with $d=2$ and $r_0=0$. To exemplify, we have $r_1=\sqrt{v_1/\omega_2}\:$, $r_2=\sqrt{(v_1+v_2)/\omega_2}$, and so on. For each $j$, we generate the position of the $j$-th point at distance $r_j$ from the pivot by sampling its angular coordinates from a uniform distribution with support $\left[0; 2\pi\right)$. \\
	\begin{figure}[ht!]
		\begin{center}
			\includegraphics[scale=.35]{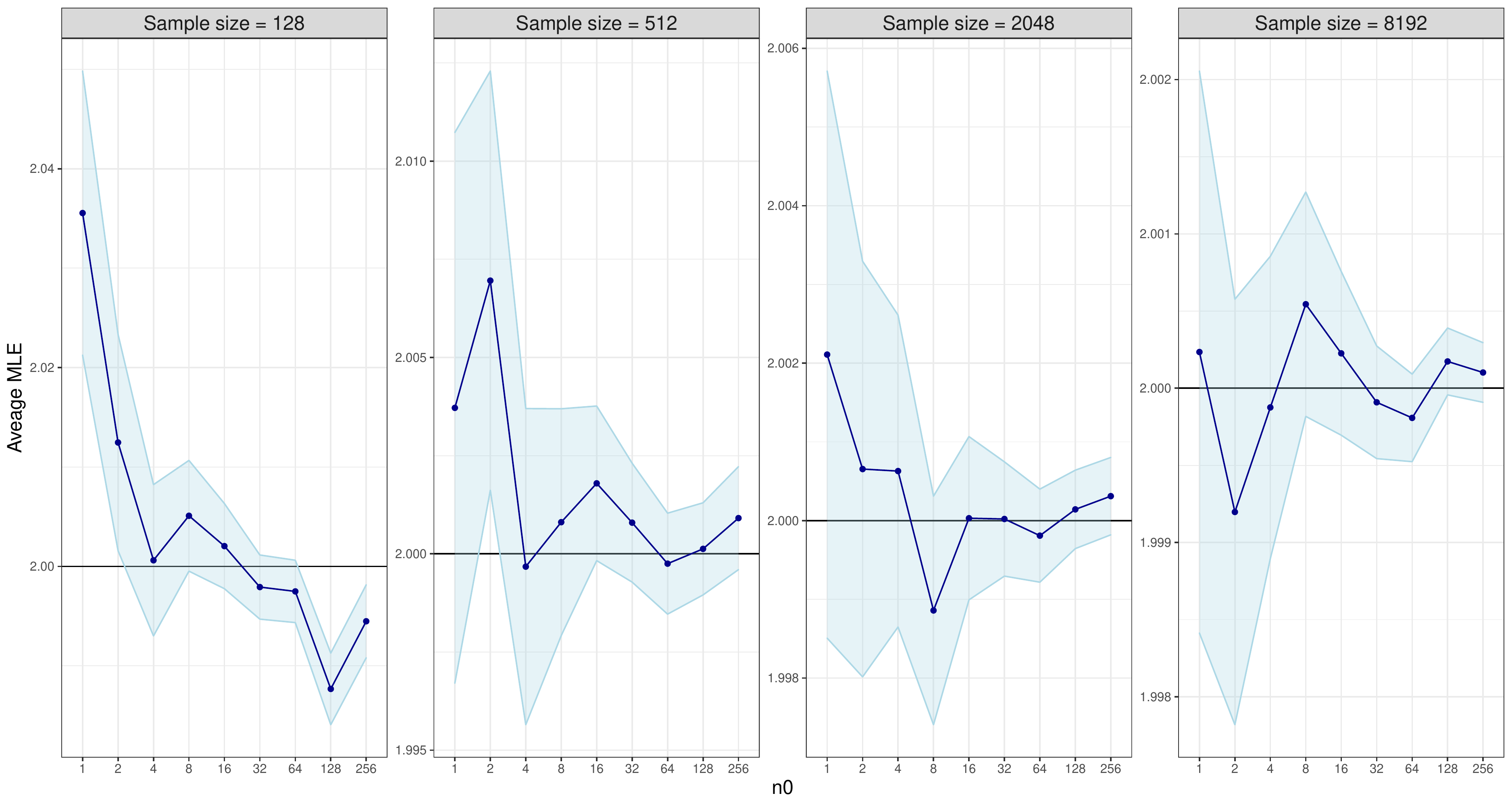}
		\end{center}
		\caption{Average MLE for the Gride models obtained over different NN orders. The different panels showcase different sample sizes considered to compute the estimates. The error bands display the $95\%$ confidence interval on the average MLE.}
		\label{fig:uniform_bootstrap}
	\end{figure}
	The panels in Figure \ref{fig:uniform_bootstrap} showcase the \texttt{id} estimates as a function of the closest $j$ NNs to the pivot, where  $j \in \{128, 512, 2048, 8192\}$. We employ different NN orders keeping the ratio ${n_2}/{n_1}=2$ fixed and we increase geometrically $n_1$ from $1$ to $256$ ($x$-axis). In this experiment, the \texttt{id} is estimated via maximum likelihood on 1,000 repeated samples. Given the sample of 1,000 estimates $\hat{d}$ we compute its average with its 95\% confidence intervals. 
	The first three panels show a small but consistent bias for the \texttt{id} estimated with $n_1=1$ (\texttt{TWO-NN}) and $n_1=2$. 
	The most viable explanation for the behavior of the estimator at small $n_1$ is the statistical correlation: the $\dot{\mu}$'s entering in the likelihood (see Equation \eqref{mudot}) are computed on nearby points, and, as a consequence, they cannot be considered purely independent realizations.  Remarkably, this correlation effect is greatly reduced when larger values of $n_1$ are considered. 
	On the other hand, the small bias we may observe at large NN orders is instead likely due to numerical error accumulation. Recall that the radii of the produced points are obtained from the sum of $l$ volumes sampled from a homogeneous Poisson process. Given the data generating mechanism we used, the statistical error might compound across the different stages.\\
	
	Despite these possible drawbacks, we can state that the family of \texttt{Gride} models provide a reliable set of estimators and an effective strategy to study how the \texttt{id} changes with the scale. 
	This second aspect is relevant for many applications. The authors in \citet{Facco} showed that a scale-dependent analysis of the \texttt{id} is essential to identify the correct number of relevant directions in noisy data and propose to decimate the dataset to increase the typical distance involved in the estimate. Instead of discarding precious information from our dataset, we here propose to apply a sequence of \texttt{Gride} models on the \emph{entire dataset} to explore larger regions: the higher $n_1$, the larger is the average neighborhood size analyzed.\\
	To investigate the impact of the scale on the \texttt{id} estimates, we simulate 50,000 data points from a two-dimensional Gaussian distribution and perturb them with orthogonal Gaussian white noise. 
	We compare the results obtained under two cases: one-dimensional (1D) and twenty-dimensional (20D) noise; in both cases, the perturbation variance is set to $\sigma^2 = 1e-4$. 
	Specifically, we estimate the \texttt{id} of the dataset with several \texttt{Gride} models by changing the ratio $n_{2, 1} = {n_2}/{n_1}$ of the order of the nearest neighbors used to compute $\dot{\mu}$. 
	The results are shown in Figure \ref{fig:gauss_crime_grime}. On the $x$-axis we report the mean neighbor distance computed as $\bar{r} = (r_{n_2}+r_{n_1})/2$, averaged over all the observations.
	\begin{figure}[th!]
		\begin{center}
			\includegraphics[scale=.45]{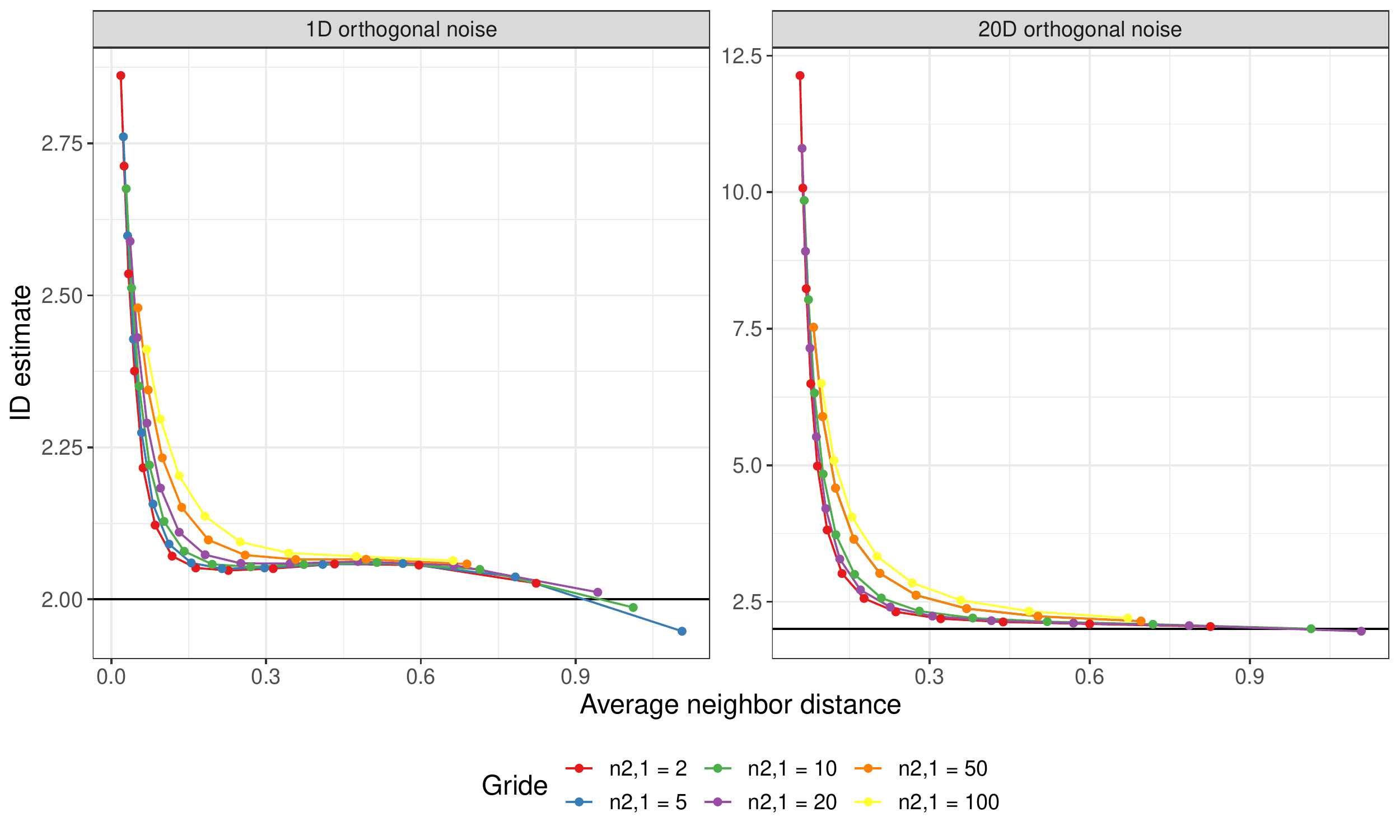}
		\end{center}
		\caption{Analysis of the impact of the scale on the \texttt{id} estimates for different \texttt{Gride} models performed on a 2D noisy Gaussian dataset. The \texttt{id} is calculated maximizing the likelihood of Equation \eqref{mudot}; the errorbars computed with the Fisher information are smaller than the marker size.}
		\label{fig:gauss_crime_grime}
	\end{figure}
	
	When $\bar{r}$ is of the same order as $\sigma$,  the \texttt{id} estimated by the \texttt{Gride} models is much higher than 2, the true value. For instance, when $\bar{r} \approx \sigma$, the geometry of the neighborhoods is approximately 3-dimensional, and, consistently, \texttt{id} $\approx 2.85$ (left panel of Figure \ref{fig:gauss_crime_grime}). 
	As we increase the range of distances involved in the estimate, all the models display a plateau around \texttt{id} $\approx 2$. However, when $n_{2, 1} =2$, the \texttt{id} stabilizes around two at smaller scales for low and high dimensional noise. Indeed, the left panel shows that  $\texttt{id} \approx 2.1$ at $\bar{r} \approx 0.08$ for $n_{2, 1} =2$, at $\bar{r} \approx 0.14$ for $n_{2, 1} =20$ and 
	$\bar{r} \approx 0.18$ for $n_{2, 1} = 50$. Similarly, in the right panel $\texttt{id} \approx 2.2$ at $\bar{r} \approx 0.2$ for $n_{2, 1} =2$, at $\bar{r} \approx 0.45$ for $n_{2, 1} =20$ and
	$\bar{r} \approx 0.6$ for $n_{2, 1} = 50$.
	A broader plateau makes it easier to identify the number of relevant directions present in the dataset. 
	Therefore, our numerical experiments suggest that the choice $n_{2, 1} =2$ is the most appropriate in practical applications.
	
	\begin{figure}[ht!]
		\begin{center}
			\includegraphics[scale=.45]{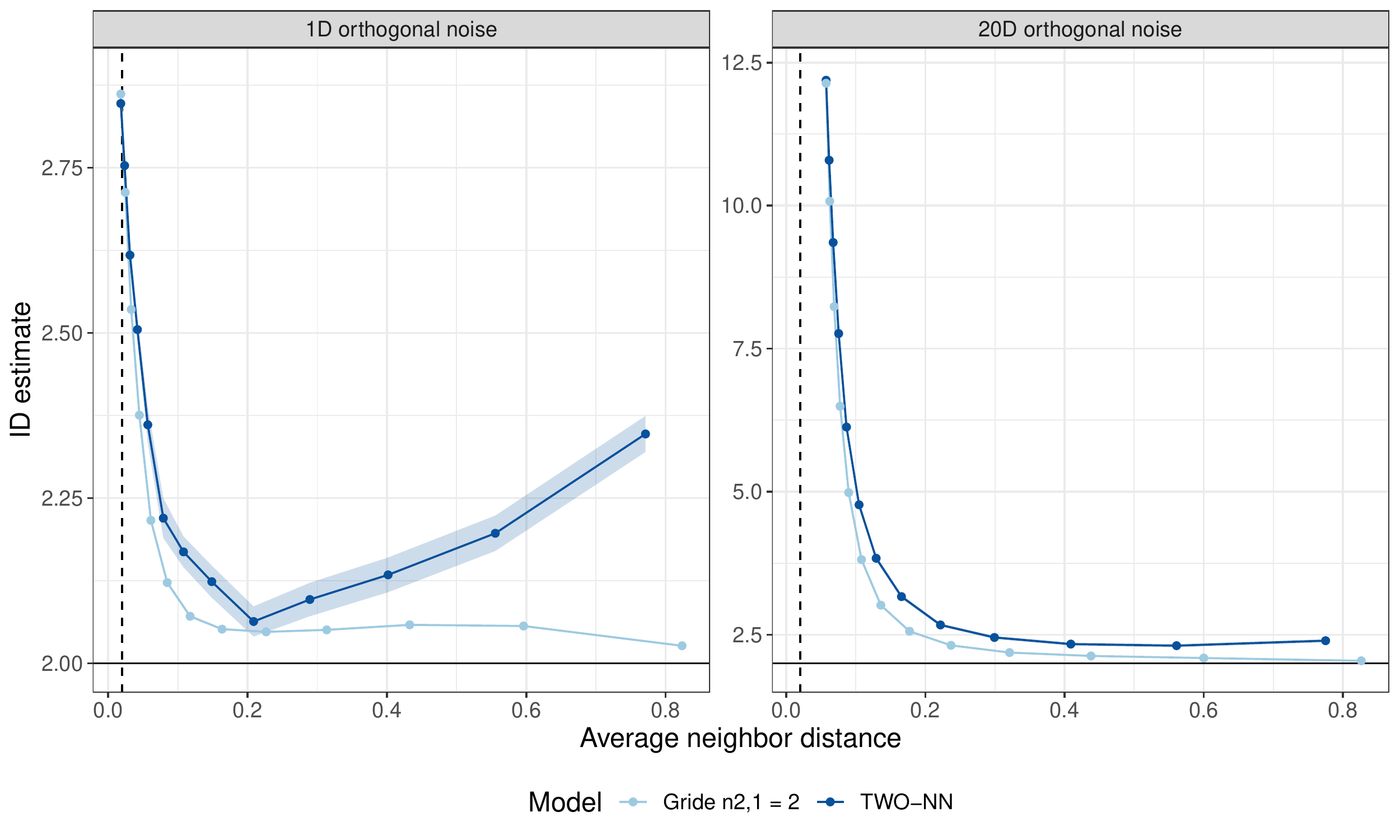}
		\end{center}
		\caption{Analysis of the impact of the scale on the \texttt{id} estimates comparing \texttt{Gride} models of order ratios $n_{2,1}=2$ vs. the \texttt{TWO-NN} estimator performed on a 2D noisy Gaussian dataset. The error bounds ($\pm 2$ std. dev.) are visible only for one model.}
		\label{fig:gauss_twonn_grime}
	\end{figure}
	Figure \ref{fig:gauss_twonn_grime} focuses on the comparison of the scale analysis done with \texttt{TWO-NN} and \texttt{Gride} with $n_{2, 1} =2$ on the same dataset. 
	Following \citet{Facco} we applied the \texttt{TWO-NN} estimator on several subsets of the original data and report the average \texttt{id} with its 95\% confidence intervals. 
	Both in the case of high and low dimensional noise \texttt{Gride} settles down to around 2 at smaller scales than the \texttt{TWO-NN} estimator. 
	The left panel also shows that the decimation protocol of \texttt{TWO-NN} can introduce a bias at large scales when the size of the replicates becomes small. 
	In our experiment, by halving the sample size at each decimation step, we use subsets with 12 datapoints when $\bar{r} \approx 0.8$. At a comparable scale, \texttt{Gride} performs much better since we always maximize the likelihood using all of the original 50,000 data points. 
	
	\section{Discussion}
	\label{Sec::Concl}
	In this paper, we introduced and developed novel distributional results concerning the homogeneous Poisson point process related with the estimate of the \texttt{id}, which is a crucial quantity for many dimensionality reduction techniques. The results extend the theoretical framework of the \texttt{TWO-NN} estimator. In detail, we derived closed-form density functions for the ratios of distances between a point and its nearest neighbours, ranked in increasing order.\\  
	
	The distributional results not only have a theoretical value per se but are also useful to improve the model-based estimation of the \texttt{id}. Specifically, we have derived two estimators: \texttt{Cride} and \texttt{Gride}. The first one builds on the independence of the elements of the vector $\{\mu_{i,l}\}_{l=1}^L$, which we exploit to derive a closed-form estimator with lower variance than the \texttt{TWO-NN}.  
	However, considering multiple ratios of distances for each point in the sample can lead to cases that violate the assumed independence of the vector of ratios $\{\bm{\mu}_{i,L}\}_{i=1}^n$, which allows us to express the likelihood as a product of marginal distributions. To mitigate this issue, we have proposed \texttt{Cride}, an estimator based on NNs of generic order. We showed that the latter estimator is also more robust to the presence of noise in the data.\\The main potential drawback of these two novel estimators when compared to \texttt{TWO-NN} is that the inclusion of NNs of higher orders has to be accompanied by stronger assumptions on the homogeneity of the density of the data-generating process. 
	Nonetheless, by dedicated computational experiments, we have shown that the assumption of homogeneity of the Poisson point process can be weakened. Indeed, given a specific point in the configuration, the homogeneity should only hold up to the scale of the distance of the furthest nearest neighbor entering the estimator.\\
	
	To summarize, when dealing with real data we face a trade-off between the assumptions of homogeneity and independence. On the one hand, the \texttt{TWO-NN} is more likely to respect the local homogeneity hypothesis but is extremely sensitive to measurement noise since it only involves a narrow neighborhood of each point. On the other hand, \texttt{Cride} focuses on broader neighborhoods, which makes it more robust to noisy data but also imposes a stronger local homogeneity requirement. It is also more likely to induce dependencies among different sequences of ratios. 
	We believe that \texttt{Gride} provides a reliable alternative to the previous two estimators, being both robust to noise and more likely to comply with the independence assumptions.\\

	
	The results in this paper pave the way for many other possible research avenues.
	First, we have implicitly assumed the existence of a single manifold of constant \texttt{id}. However, it is reasonable to expect that a complex dataset can be characterized by multiple latent manifolds with heterogeneous $\id$s. \citet{Allegra} extended the \texttt{TWO-NN} model in this direction by proposing \texttt{Hidalgo}, a tailored mixture of Pareto distributions to partition the data points into clusters driven by different $\id$ values. It would be interesting to combine the \texttt{Hidalgo} modeling framework with our results, where the distributions in Equations \eqref{musgammas} and \eqref{mudot} can replace the Pareto mixture kernels.
	Second, the estimators derived from the models do not directly consider any source of error in the observed sample. Although we showed how one can reduce the bias generated by this shortcoming by considering higher-order nearest neighbours that allow escaping the local distortions, we are still investigating how to address this issue more broadly. For example, a simple solution would be to model the measurement errors at the level of the ratios, accounting for a gaussian noise that can distort each $\mu_i$. 
	A more promising solution to this problem may be given by the Generalized Gamma distribution derived in Theorem \ref{Lemma4}.
	By focusing directly on the distribution of the distances between data points in an ideal, theoretical setting, we can obtain informative insights on how to best model the measurement noise.

	
	
	\bibliographystyle{plainnat} 
	\bibliography{IDbiblio}       

	\section{Appendix}
	\subsection{Proof of Theorem \ref{Lemma4}}
	\begin{proof}
		To simplify the notation, let us drop the subscript $i$. Recall that we were able to prove that $f(v_1,\ldots,v_L)=\prod_{l=1}^{L}f(v_l)$, where $f(v_l)=\rho\exp(-\rho v_l)$, meaning that $v_l\stackrel{i.i.d.}{\sim}Exp(\rho)$. We consider the following one-to-one transformation for $l=1,\ldots,L$:
		$$ r_l = \left( \frac{\sum_{k=1}^{l}v_k}{\omega_d} \right)^{1/d} \iff v_l = \omega_d \left(r_l^d-r_{l-1}^d\right). $$
		The determinant of the Jacobian of this transformation is
		$|J|= (\omega_d d)^L\prod_{l=1}^{L}r_l^{d-1}$.
		Thus, the distribution of the first $L$ distances has density:
		
		$$ f(r_1,\ldots,r_L)= (\rho \omega_d d)^L  \left(\prod_{l=1}^{L}r_l^{d-1} \right)\exp\left[ -\rho\omega_d r_L^d \right], $$
		with $r_l \in \mathbb{R}^+$ and the constraint that $r_1<r_2<\ldots<r_L$.\\
		We can also derive the marginal distribution of the generic distance $r_L$. This can be easily done by repeatedly integrating out the smallest distance $r_l$ over $\left(0,r_{l+1}\right)$, $l=1,\ldots,L-1$. In formulas:
		
		\begin{align*}
			f(r_L)&=(\rho \omega_d d)^Lr_L^{d-1}\exp\left[ -\rho\omega_d r_L^d \right]\int_{0}^{r_{L}}\int_{0}^{r_{L-1}}\cdots\int_{0}^{r_{2}}  \left(\prod_{l=1}^{L-1}s_l^{d-1} \right) ds_1 \cdots ds_{L-1}\\
			&=\exp\left[ -\rho\omega_d r_L^d \right] (\rho \omega_d d)^L \frac{r_L^{Ld-1}}{(L-1)!d^{L-1}}.
		\end{align*}
		We conclude that the generic distance from a point to its $L$-th NN follows a Generalized Gamma distribution, whose density is given by
		$$f(x) ={\frac  {p/a^{q}}{\Gamma (q/p)}}x^{{q-1}}e^{{-(x/a)^{p}}},\quad\quad x,a,p,q>0. $$
		Therefore, $f(r_L)$ is a Generalized Gamma density with parameters $p=d,\quad a=\frac{1}{\sqrt[d]{\rho \omega_d}},\quad q=Ld$.	
		There is another, faster way to recover this last result.
		Since $v_l\sim Exp(\rho)$ for each $l=1,\ldots,L$, it is easy to see that the volume of the hyper-sphere of radius $r_L$, defined as $V_L=\sum_{l=1}^{L}v_l=\omega_d r_L^d$ follows an Erlang distribution: $V_L\sim Gamma(L,\rho)$.
		Then,
		$$ V_L\sim Gamma(L,\rho) \iff r_L^{d} = \frac{V_L}{\omega_d}\sim Gamma(L,\omega_d\rho) \iff r_L\sim GenGamma\left(d,\frac{1}{\sqrt[d]{\rho \omega_d}},Ld\right). $$
	\end{proof}

\end{document}